\documentclass{scrartcl}
\usepackage[T1]{fontenc}
\usepackage{lmodern}

\usepackage{authblk}

\usepackage{amsmath}
\usepackage{amssymb}
\usepackage{amsthm}
\usepackage{mathtools}
\usepackage{enumitem}
\usepackage{xcolor}
\usepackage{quiver}

\usepackage{microtype}

\usepackage{hyperref}
\setkomafont{author}{\normalsize}

\newtheorem{theorem}{Theorem}
\newtheorem{lemma}{Lemma}
\newtheorem{definition}{Definition}

\title{Extremal Fitting CQs do not Generalize}
\author[1]{Balder ten Cate}
\affil[1]{ILLC, University of Amsterdam}
\author[2]{Maurice Funk}
\author[3]{Jean Christoph Jung}
\affil[3]{TU Dortmund University}
\author[2]{Carsten Lutz}
\affil[2]{Leipzig University and ScaDS.AI}

\date{}

\begin{document}

\newcommand{\Amf}{\mathbf{A}}
\newcommand{\Dmc}{\mathcal{D}}
\newcommand{\Amc}{\mathcal{A}}
\newcommand{\Bmc}{\mathcal{B}}
\newcommand{\Smc}{\mathcal{S}}
\newcommand{\mn}[1]{\ensuremath{\mathsf{#1}}}

\maketitle{}

\begin{abstract}
  A fitting algorithm for conjunctive queries (CQs) produces, given a set of   
  positively and negatively labeled data examples, a CQ that fits 
  these examples. In general, there may be many non-equivalent fitting CQs and
  thus the algorithm has some freedom in producing its output. Additional
  desirable properties of the produced CQ are that it generalizes well to
  unseen examples in the sense of PAC learning and that it is most general
  or most specific in the set of all fitting CQs. In this research note, we
  show that these desiderata are incompatible when we require 
  PAC-style generalization from a \emph{polynomial} sample: we
  prove that any fitting algorithm
  that produces a most-specific fitting CQ cannot be a sample-efficient PAC
  learning algorithm, and the same is true for fitting algorithms that
  produce a most-general fitting CQ (when it exists). Our proofs
  rely on a polynomial construction of relativized homomorphism 
 dualities for path-shaped structures.
%
%
\end{abstract}

\section{Introduction}

The classic \emph{query-by-example} paradigm~\cite{zloof1975query}, also known as
\emph{query reverse engineering}, aims to assist users in query formation. At its heart lies the \emph{fitting problem}, which is the task to
compute a query that fits a given set of (positively and negatively) labeled data examples. The fitting problem has been
intensively studied for the  class of conjunctive queries, see for example
(CQs)~\cite{Tran2014:reverse,Li2015:qfe,Barcelo017}.

An algorithm that solves this task, that is, takes as input a set of labeled
data examples and returns a CQ that fits the examples, is called a fitting algorithm.
As most sets of examples have multiple fitting CQs, the returned CQ depends on
the chosen fitting algorithm.
If one desires the fitting CQ to have additional desirable properties, one might
thus choose a fitting algorithm that is tailored towards this purpose. 
See \cite{sigmod-to-appear} for an overview.

In many query reverse engineering scenarios, one wants a fitting algorithm that
produces a CQ that 
correctly labels future examples, provided that they are
drawn from the same distribution as the examples given as an input. This property is formalized by Valiant's
framework of probably approximately correct (PAC) learning~\cite{Valiant84:pac}. Given
sufficiently many examples sampled from an unknown distribution, a PAC learning
algorithm is required to return  with high probability $1 - \delta$ a CQ that has small
classification error $\epsilon$, for $\delta,\epsilon \in (0,1)$.
Additionally, it is desirable that the number of required examples only grows
polynomially with the parameters  $\delta,\epsilon$. If this is the case, we say that the PAC
learning algorithm is \emph{sample-efficient}.


From a logical perspective, extremal fitting CQs are of particular interest.
These are fitting CQs that are most general or most specific in the class
of all fitting CQs. When they exist,  extremal fitting CQs describe the entire
(potentially infinite) set of fitting CQs and thus they provide a useful basis
to understand this set. Extremal fitting CQs have recently been studied in
detail in \cite{pods2023:extremal}. 

A natural question is how the desiderata of PAC learning and of extremality
relate. In particular, one may ask whether there is a fitting algorithm that 
is both a sample efficient PAC learning algorithm and produces a most-specific or most-general
fitting CQ, when it exists. This question was answered to the negative in \cite{IJCAI23:SAT} for 
a class of tree-shaped
CQs, there called $\mathcal{EL}$ queries, both for most-specific and for most-general fittings.\footnote{We remark that this is not always the case. In the case of monomials
(conjunctions of propositional literals), for instance, fitting algorithms that produce the shortest fitting monomial
are sample-efficient PAC learning algorithms~\cite{anthonyComputationalLearningTheory1992}
and shortest fitting monomials are also most general fittings.}
The aim of this note is to extend these results to unrestricted CQs and to present them in a
self-contained way. More specifically, we prove the following:
\begin{enumerate}

  \item Any fitting algorithm for CQs that always produces a most-general fitting CQ, when it exists, is not a sample-efficient PAC learning algorithm.

  \item Any fitting algorithm for CQs that always produces a most-specific fitting CQ is not a sample-efficient PAC learning algorithm.

\end{enumerate}

\section{Preliminaries}
A schema $\Smc$ is a set of relation symbols, each with associated arity.
A \emph{database instance} over $\Smc$ is a finite set $I$ of
\emph{facts} of the form $R(a_1,\dots,a_n)$ where $R \in \Smc$ is a
relation symbol of arity $n$ and $a_1,\dots,a_n$ are \emph{values}.
We use $\mn{adom}(I)$ to denote the set of all values used in~$I$.

A $k$-ary \emph{conjunctive query} (CQ) over a schema $\Smc$ for $k \geq 0$, is
an expression of the form
\[
  q(\textbf{x}) \text{ :- } \alpha_1,\ldots,\alpha_n
\]
where $\textbf{x}=x_1, \ldots, x_k$ is a sequence of variables and each
$\alpha_i$ is a relational atom that uses a relation symbol from $\mathcal{S}$
and no constants. 
The variables in $\textbf{x}$ are called \emph{answer variables} and
the other variables used in the atoms~$\alpha_i$ are the \emph{existential variables}.
We denote by $q(I)$ the set of all $k$-tuples $\mathbf{a} \in \mn{adom}(I)^k$ such that $I \models
q(\mathbf{a})$.
With the \emph{size} $||q||$ of a CQ~$q$, we mean
the number of atoms in it.
For CQs $q_1$ and $q_2$ of the same arity and over the same schema
$\Smc$, we write $q_1 \subseteq q_2$ and say that $q_1$ is
\emph{contained} in $q_2$ if $q_1(I) \subseteq q_2(I)$ for all database
instances $I$ over $\Smc$. 
If $q_1 \subseteq q_2$ and $q_2 \subseteq q_1$, then we call $q_1$ and $q_2$
\emph{equivalent} and write $q_1 \equiv q_2$.

A \emph{data example} $(I, \mathbf{a})$ for a $k$-ary CQ $q$ consists of a database instance $I$
together with a tuple $\mathbf{a} \in \mn{adom}(I)^k$.
A data example is
\begin{itemize}
  \item a \emph{positive example} for $q$ if $\textbf{a}\in q(I)$; 
  \item a \emph{negative example} for $q$ if $\textbf{a}\notin q(I)$.
\end{itemize}
A collection of labeled examples is a pair $E = (E^+, E^-)$ where $E^+$ and
$E^-$ are sets of examples. We say that a CQ \emph{fits} $E$ if each element of
$E^+$ is a positive example for $q$ and each element of $E^-$ is a negative
example for $q$.
A \emph{fitting algorithm} takes as input a collection of labeled examples $E$
for which a fitting CQ is promised to exist, and returns such a CQ.

A \emph{homomorphism} $h$ from an instance $I$ to an instance $J$
(over the same schema) is a function $h \colon \mn{adom}(I) \to \mn{adom}(J)$ such that
the $h$-image of every fact of $I$ is a fact of $J$.
We write $I \to J$ to indicate the existence of a homomorphism from $I$ to $J$.
For two data examples $(I, \mathbf{a})$ and $(J, \mathbf{b})$ we write
$(I, \mathbf{a}) \to (J, \mathbf{b})$ if there is a homomorphism
$h$ from $I$ to $J$ with $h(\mathbf{a}) = \mathbf{b}$.

The \emph{canonical database instance} of a CQ $q$ is the instance $I_q$ (over the
same schema as $q$) whose active domain consists of the variables that occur in
$q$ and whose facts are the atomic formulas in $q$. The \emph{canonical example} $e_q$ of a CQ $q(\mathbf{x})$
is the example $(I_q, \mathbf{x})$.
The \emph{canonical CQ} $q_I(\mathbf{a})$ of an example $(I, \mathbf{a})$ is the
CQ whose atoms are the facts of $I$ and whose variables are the values
in $\mn{adom}(I)$.

It is well known that for all CQs $q(\mathbf{x})$ and instances $I$ over the
same schema, $\mathbf{a} \in q(I)$ if and only if $(I_q, \mathbf{x}) \to (I,
\mathbf{a})$. Furthermore, for all CQs
$q_1(\mathbf{x})$ and $q_2(\mathbf{y})$ of the same arity and over the same
schema, $q_1 \subseteq q_2$ if and only if $(I_{q_2}, \mathbf{y}) \to (I_{q_1},
\mathbf{x})$. We will use these facts freely.

\subsection{PAC Learning}

For CQs $q$, $q_T$ of the same arity $n$ and a probability distribution over examples $P$ of arity~$n$, we define
\[\mn{error}_{P,q_T}(q)= \mathop{\textrm{Pr}}_{(I, \mathbf{a}) \in P}(\mathbf{a} \in q(I) \mathop{\triangle} q_T(I))\]
to be the expected
error of $q$ relative to $q_T$ and $P$ where $\triangle$ denotes symmetry difference.
In other words, $\mn{error}_{P,q_T}(q)$ is the probability that $q$ disagrees with $q_T$
on an example drawn at random from $P$.

\begin{definition}
  \label{def:efflearn}
  A \emph{PAC learning algorithm} is a (potentially randomized) algorithm $\Amf$
  associated with a function
  $m \colon \mathbb{R}^2 \times \mathbb{N}^3 \rightarrow \mathbb{N}$ such
  that
  \begin{itemize}

  \item $\Amf$ takes as input a collection of labeled data examples $E$;

  \item for all $\epsilon,\delta \in (0,1)$, 
    all finite schemas $\Smc$, all $s_Q,s_E \geq 0$, all
    probability distributions $P$ over examples $(I, \mathbf{a})$
    over $\Smc$ with $|I| \leq s_E$,
    and all CQs $q_T$ with $||q_T|| \leq s_Q$, the
    following holds: when running $\Amf$ on a collection $E$ of
    at least $m(\frac{1}{\delta},\frac{1}{\epsilon},|\Smc|, s_Q, s_E)$
    data examples drawn from $P$ and labeled by $q_T$, 
    $\Amf$ returns a CQ $q_H$ such that with
    probability at least $1 - \delta$ (over the choice of $E$), we
    have $\mn{error}_{P,q_T}(q_H) \leq \epsilon$.
\end{itemize}
We say that $\Amf$ \emph{has sample size} $m$ and call $\Amf$
\emph{sample-efficient} if $m$ is a polynomial.
\end{definition}

Some stricter definitions of PAC learning algorithms let $m$ only depend on
$\frac{1}{\delta}$ and $\frac{1}{\epsilon}$. By including $|\Smc|$, $s_E$ and
$s_Q$ as parameters of $m$, our definition is more liberal and we obtain
stronger lower bounds.
Note that we do not impose any requirements on the running time of a PAC learning
algorithm,  and that sample efficiency only requires the  number of examples 
needed to achieve generalization to be bounded polynomially, but not the running
time.
This is because it is known that there are no polynomial time PAC learning
algorithms for CQs unless RP$=$NP~\cite{IPL24:PAC,DBLP:conf/ecml/Kietz93}.

We remark that sample-efficient PAC learning algorithms exist. In fact,
using a well known connection between PAC learning algorithms and Occam
algorithms~\cite{DBLP:journals/jacm/BlumerEHW89}, one can obtain the following.
\begin{theorem}[\cite{IPL24:PAC}]
  Let $\Amf$ be a fitting algorithm that always produces the smallest fitting CQ.
  Then $\Amf$ is a sample-efficient PAC learning algorithm.
\end{theorem}
An algorithm of this kind can easily be obtained:
enumerate all CQs over $\Smc$ in order of increasing size and check for each
candidate $q$, whether it fits the input examples. For some classes of queries,
such algorithms can even be efficiently implemented \cite{IJCAI23:SAT}.

\subsection{Extremal Fittings}

Besides minimal size, there are other interesting properties that a fitting CQ
might have. From a logical perspective, fitting CQs stick out that are \emph{extremal}
in the sense of being most general or most specific among all fitting CQs.
For example, extremal fittings are interesting because the set of all fitting CQs is convex,
meaning that if $q_1$ and $q_2$ are fitting CQs and $q_1 \subseteq q \subseteq q_2$,
then $q$ is also a fitting CQ. When they exist, extremal fitting CQs thus characterize
the space of all fitting CQs. They have recently been studied in detail in \cite{pods2023:extremal}. 
%

Let $E$ be a collection of labeled examples. A CQ $q$ is a
\begin{itemize}
  \item \emph{most-specific fitting} CQ for $E$ if $q$ fits $E$ and for every CQ $q'$ that fits $E$,
    $q \subseteq q$;
  \item  \emph{strongly most-general fitting} CQ for $E$ if $q$ fits $E$ and
    for every CQ  $q'$ that fits $E$, $q'\subseteq q$; 
  \item \emph{weakly most-general fitting} CQ for  $E$ if $q$ fits $E$ and
    for every CQ $q'$ that fits $E$, $q\subseteq q'$ implies $q \equiv q'$.
\end{itemize}
There is also a more general notion than strongly most-general fitting:
\begin{itemize}
    \item
    a finite set of CQs $\{q_1, \ldots, q_n\}$ is a \emph{basis of most-general fitting CQs}
for  $E$ if each $q_i$ fits $E$ and 
for all CQs $q'$ that fit $E$, we have $q'\subseteq q_i$ for some $i\leq n$.
\end{itemize}
Note that the definition of most-specific fitting CQs parallels that of strongly most-general fitting CQs. One may also define a weak version of most-specific fitting CQs,
but it turns out to be equivalent \cite{pods2023:extremal}.
The different notions of
most-general fitting CQ are connected as follows.
If $q$ is a strongly most-general fitting CQ, then $\{q\}$ is a basis of
most-general fitting CQs. If a basis of most-general fitting CQs is subset minimal,
then it contains only weakly most-general fitting CQs.
Furthermore, if a strongly most-general fitting CQ exists, then it is the only
weakly most-general fitting CQ.

A most-specific fitting CQ $q$ and a basis of most-general fitting CQs $\{q_1,
\dots, q_n\}$ completely describe the space of all fitting CQs: For all CQs
$q'$, $q \subseteq q'$ and $q' \subseteq q_i$ for some $i$, if and only if $q'$
fits $E$.

Most-specific fitting CQs are closely related to direct products of examples.
The \emph{direct product} $I \times J$ of two instances $I$ and $J$ is the instance that consists of
all facts $R(\langle a_1, b_1\rangle, \ldots, \langle a_n,b_n\rangle)$,
where $R(a_1, \ldots, a_n)$ is a fact in $I$ and 
$R(b_1, \ldots, b_n)$ is a fact in~$J$. 
The direct product $(I,\textbf{a})\times (J,\textbf{b})$ of two $k$-ary
examples, where $\textbf{a}=a_1, \ldots, a_k$ and $\textbf{b}=b_1, \ldots b_k$,
is $(I\times J, (\langle a_1, b_1\rangle, \ldots, \langle a_k,b_k\rangle))$.
\begin{theorem}[\cite{pods2023:extremal}]\label{thm:most-specific-fitting-chara}
For all CQs $q$ and collections of labeled examples $E=(E^+,E^-)$, the following are equivalent:
\begin{enumerate}
    \item  $q$ is a most-specific fitting CQ for $E$,
    \item  $q$ fits $E$ and is equivalent to the canonical CQ of $\Pi_{e\in E^+}(e)$.
\end{enumerate}
\end{theorem}

Strongly most-general fitting CQs and, more generally, finite bases of
most-general fitting CQs are closely related to \emph{homomorphism dualities}, a
fundamental concept that originates from combinatorial graph theory and that has
found diverse applications in different areas, including the study of constraint
satisfaction problems, database theory and knowledge representation. Here, we
use a relativized version of homomorphism dualities introduced in \cite{pods2023:extremal}.

\begin{definition}
Let $(J, \mathbf{b})$ be an example. A \emph{homomorphism duality relative to}
$(J, \mathbf{b})$ is a pair of finite sets of examples $(\mathcal{F}, \mathcal{D})$ such that for
all examples $(J', \mathbf{b}')$ with $(J', \mathbf{b}') \to (J, \mathbf{b})$,
it holds that $(I_f, \mathbf{a}_f) \to (J', \mathbf{b}')$ for some $(I_f,
\mathbf{a}_f) \in \mathcal{F}$ if and only if $(J', \mathbf{b}') \not\to (I_d,
\mathbf{a}_d)$ for all $(I_d, \mathbf{a}_d) \in \mathcal{D}$.
\end{definition}

\begin{theorem}[\cite{pods2023:extremal}]\label{thm:most-general-fitting-chara}
For all CQs $q_1, \ldots, q_n$ and collections of labeled examples $E$, the 
following are equivalent:
\begin{enumerate}
    \item $\{q_1, \ldots, q_n\}$ is a basis of most-general fitting CQs for $E$
    \item each $q_i$ fits $E$, and 
    $(\{e_{q_1}, \ldots, e_{q_n} \}, E^-)$ is a homomorphism duality relative to $\Pi_{e\in E^+} (e)$.
\end{enumerate}
\end{theorem}
Or, formulated for strongly most-general fitting CQs:
$q$ is a strongly most-general fitting CQ if and only if $q$ fits $E$ and
$(\{e_q\}, E^-)$ is a homomorphism duality relative to $\Pi_{e\in E^+} (e)$.

\section{Most-General Fittings Preclude Sample-Efficient PAC Learning}

We show that fitting algorithms that always produce a most-general fitting CQ,
if it exists,
cannot be sample-efficient PAC learning algorithms. We focus on strongly
most-general fitting CQs since fitting algorithms that always produce
a weakly most-general fitting CQ or the elements of a base of most-general fitting CQs
must also produce a strongly most-general fitting CQ, if it exists. Consequently,
our result also applies to the latter kinds of algorithms.

The existing similar results in \cite{IJCAI23:SAT}
that pertain to tree-shaped CQs  rely heavily on
\emph{simulation dualities}, defined like homomorphism dualities but using the
weaker simulation relation instead of homomorphisms.
In fact, \cite{pods2023:extremal} gives a characterization of bases of most-general fitting tree CQs in terms of simulation dualities instead of
homomorphism dualities.
Guided by Theorem~\ref{thm:most-general-fitting-chara}, we aim to base
our proofs on homomorphism
dualities instead.

Known constructions of a duality $(\mathcal{F}, \mathcal{D})$ from $\mathcal{F}$
relative to some example $(J, \mathbf{b})$ result in
examples $\mathcal{D}$ that have  size exponential in $||\mathcal{F}||$ and
$|J|$~\cite{tCD2022:conjunctive}.
This makes them unsuitable to show non-sample-efficiency as it would force
us to use negative data examples that are exponentially larger than the
positive examples. Since Definition~\ref{def:efflearn} allows the sample size $m$ to depend on
the size of the examples, this would effectively allow a
sample-efficient PAC learning algorithm to use an exponential number of examples.
In order to avoid this effect, we begin by showing that for a certain
restricted class of examples $\mathcal{F}$, a homomorphism duality $(\mathcal{F}, \mathcal{D})$ exists
such that  $\mathcal{D}$ is of polynomial size 
(and  can be computed in polynomial time).

A \emph{path instance} is a database instance $I$ such that 
the facts in $I$ are of the form
\[
    R_1(a_0, a_1), \dots, R_n(a_{n - 1}, a_n), P_1(a_{j_1}), \dots, P_m(a_{j_m})
\]
where all $R_i$ are binary relation symbols, all $P_i$ are unary relation
symbols and for all $j_i$, $j_i > 0$.
A \emph{path example} $(I, a)$ is a unary data example where $I$ is a path instance and $a =
a_0$.

\begin{lemma}\label{lem:path-dual}
Let $(I, a_0)$ and $(J, b_0)$ be path examples. There exists a data example $(D,
d)$ that can be computed in time polynomial in $|I|$ and $|J|$ such that
$(\{(I, a_0)\}, \{(D, d)\})$ is a homomorphism duality relative to $(J, b_0)$.
\end{lemma}

\begin{proof}
  Since $I$ and $J$ are path instances, assume that $\mn{adom}(I) = \{a_0, \dots, a_n\}$
  and $\mn{adom}(J) = \{b_0, \dots b_m\}$.
  We can check in polynomial time whether $(I, a_0) \to (J, b_0)$. We distinguish cases.

  If $(I, a_0) \not\to (J, b_0)$, then
  it follows that
  $(I, a_0)  \not\to (J', b')$ 
  for all data examples $(J', b')$ with $(J', b') \to (J, b_0)$.
  Hence, $(\{(I, a_0)\}, \{(J, b_0)\})$ is a homomorphism duality relative to $(J,
  b_0)$.

  If $(I, a_0) \to (J, b_0)$, we construct a data example $(D, d)$ such
  that $(\{(I, a_0)\}, \{ (D, d) \})$ is the desired homomorphism duality relative to $(J,
  b_0)$ as follows.
  Since $(I, a_0) \to (J, b_0)$, it must be that $n \leq m$ and the binary relation symbol $R_i$
  is the same in $I$ and $J$ for all $i$ with $1 \leq i \leq n$.
  The values used in $D$ will be pairs $\langle b_i, f\rangle$ where $b_i \in \mn{adom}(J)$ and
  $f$ is a fact from $I$ that mentions $a_i$ or the dummy fact $\circ$.
  More specifically, $\mn{adom}(D)$ will be the following set of
  values:
  \begin{align*}
    & \{ \langle b_i, \circ \rangle \mid i = 0 \text{ or } n < i \leq m \} \cup {} \\
    & \{ \langle b_i, R_i(a_{i - 1}, a_i) \rangle \mid 1 \leq i \leq n \} \cup {} \\
    & \{ \langle b_i, R_{i + 1}(a_i, a_{i + 1}) \rangle \mid 0 \leq i \leq n - 1 \} \cup {} \\
    & \{ \langle b_i, P(a_i) \rangle \mid 1 \leq i \leq n \text{ and } P(a_i) \in I \}.
  \end{align*}
  For all values $\langle b_i, f \rangle, \langle b_{i + 1}, f'\rangle$ in the above set we add to $D$  the fact
  \begin{itemize} 
  \item $R_{i + 1}( \langle b_i, f \rangle, \langle b_{i + 1}, f' \rangle)$ if 
  $f \neq R_{i + 1}(a_i, a_{i + 1})$ or $f' \neq R_{i + 1}(a_{i}, a_{i + 1})$,
  \item  $P(\langle b_{i + 1}, f' \rangle)$ if $P(b_{i + 1}) \in J$ and $f' \neq P(a_{i + 1})$.
  \end{itemize}
  Then set $d$ to be the value $\langle b_0, R_1(a_0, a_1) \rangle$. 
  This completes the construction of $(D, d)$.
  Note that every binary fact $R_{i + 1}(\langle b_i, f\rangle, \langle b_{i +
  1}, f'\rangle)$ in $D$ is a copy of a binary fact $R_{i + 1}(b_i, b_{i + 1})$
  in $J$. Since $f, f' \in I \cup \{\circ\}$, there are at most $(|I| + 1)^2$
  copies of every binary fact of $J$ in $D$. The same is true for unary facts.
  Therefore $|D| \leq |J| \cdot (|I| + 1)^2$. It is easy to
  see that $D$ can be computed in polynomial time using the above construction.

  \medskip

  It remains to show that $(\{(I, a_0)\}, \{(D, d)\})$ is a homomorphism duality relative to $(J, b_0)$.
  First, we show that $(I, a_0) \to (J', b')$ implies $(J', b') \not\to (D,
  d)$ for all examples $(J', b')$ with $(J', b') \to (J, b_0)$. For this, it suffices to show that
  $(I, a_0) \not \to (D, d)$. 
  Let $I_k$ be the restriction of $I$ to facts that 
  contain only values $a_i$ with $i \geq k$. We show that
  for all $i$ with $1 \leq i \leq n$ and all $\langle b_i, f\rangle \in \mn{adom}(D)$, 
  \[
    (I_i, a_i) \to (D, \langle b_i, f \rangle)\quad  \text{implies}\quad  f = R_i(a_{i - 1}, a_i). \tag{$*$}
  \]
  We show this by induction on $n - i$. In the induction start, let $i = n$.
  By construction of $D$, $f \in \{R_n(a_{n - 1}, a_n)\} \cup \{ P(a_n) \mid
  P(a_n) \in I \}$.
  Assume that $(I, a_n) \to (D, \langle b_n, f \rangle)$. Then it follows that
  $P(\langle b_n, f\rangle) \in D$ for all $P(a_n) \in I$, hence $f \neq
  P(a_n)$ for any $P(a_n) \in I$. Therefore, the only possibility that remains is
  that $f = R_n(a_{n - 1}, a_n)$.

  In the induction step, assume that the implication holds for $i + 1$ and that
  $(I_i, a_i) \to (D, \langle b_i, f \rangle)$.
  Again, this implies that $P(\langle b_i, f \rangle) \in D$ for all $P(a_i) \in I$, 
  hence $f \notin \{ P(a_i) \mid P(a_i) \in I\}$.
  Additionally, from $(I_i, a_i) \to (D, \langle b_i, f\rangle)$ and $i < n$ it follows that 
  there must be a fact $R_{i + 1}(\langle b_i, f \rangle, \langle b_{i + 1}, f' \rangle) \in D$
  such that $(I_{i + 1}, a_{i + 1}) \to (D, \langle b_{i + 1}, f' \rangle)$.
  By the induction hypothesis, $f' = R_{i + 1}(a_i, a_{i + 1})$.
  It follows from $R_{i + 1}(\langle b_i, f \rangle,
  \langle b_{i + 1}, f' \rangle) \in D$, that $f \neq R_{i + 1}(a_i, a_{i +
  1})$. Therefore, the only possibility that remains is that $f = R_{i}(a_{i - 1}, a_i)$.

  From ($*$) and $R_1(a_0, a_1) \in I$ it now follows that
  $(I, a_0) \not\to (D, \langle b_0, R_1(a_0, a_1) \rangle) = (D, d)$, since by construction of $D$, $R_1(\langle b_0,
  R_1(a_0, a_1)\rangle, \langle b_1, R_1(a_0, a_1) \rangle) \notin D$.

  \medskip

  Second, we show that $(I, a_0) \not\to (J', b')$ implies $(J', b') \to (D, d)$
  for all examples $(J', b')$ with $(J', b') \to (J, b_0)$.
  Let $(J', b')$ be an example such that
  $(I, a_0) \not\to (J', b')$ and $(J', b') \to (J, b_0)$.
  Let $h$ be a homomorphism from $J'$ to $J$ with $h(b') = b_0$.
  We construct a homomorphism $g$ from $J'$ to $(D, d)$
  with $g(b') = d$ as follows.
  For all $a \in \mn{adom}(J')$, there is an $i$ such that $0 \leq i \leq m$
  and $h(e) = b_i$. Define $g(a)$ depending on this $i$ as follows.
  For $i = 0$,
  \begin{itemize}
    \item if $(I_i, a_i) \to (J', a)$, set $g(a) = \langle b_0, \circ \rangle$;
    \item otherwise, that is, if $(I_i, a_i) \not\to (J', a)$, set $g(a) = \langle b_0, R_1(a_0, a_1) \rangle$.
  \end{itemize}
  For $1 \leq i \leq n$,
  \begin{itemize}
    \item if $(I_i, a_i) \to (J', a)$, set $g(a) = \langle b_i, R_i(a_{i - 1}, a_i) \rangle$;
    \item otherwise, that is, if $(I_i, a_i) \not\to (J', a)$, we distinguish cases:
    \begin{itemize}
    
        \item  if there is a $P(a_i) \in I_i$ such that $P(a) \notin J'$, set $g(a) = \langle b_i, P(a_i) \rangle$;
        
        \item otherwise, set $g(a) = \langle b_i, R_{i + 1}(a_i, a_{i
    + 1}) \rangle$ (note that in this case $i < n$ and $(I_{i
    + 1}, a_{i + 1}) \not\to (J', a')$ for all $R_{i + 1}(a, a') \in J'$).
    \end{itemize}
  \end{itemize}
  For $i > n$, set $g(a) = \langle b_i, \circ \rangle$.

  It remains to verify that $g$ is a homomorphism.
  Let $P(a)$ be a unary fact in $J'$. 
  Since $h$ is a homomorphism, there is a fact $P(b_i) \in J$ with $h(a) = b_i$. 
  By definition of $g$, $g(a) = \langle b_i, f \rangle$ for
  some fact $f \in \{ \circ, R_{i}(a_{i - 1}, a_i), R_{i + 1}(a_{i}, a_{i + 1}) \} \cup \{ P'(a_i) \in I \mid P' \neq P \}$.
  From the construction of $D$ it follows that $P(g(a)) \in D$ for all these
  cases of $f$.

  Let $R(a, a')$ be a binary fact in $J'$. Since $h$ is a
  homomorphism, there is a fact $R(b_i, b_{i + 1}) \in J$ with $h(a) = b_i$
  and $h(a') = b_{i + 1}$. The mapping $g$ then maps $a$ to $\langle b_i, f
  \rangle$ and $a'$ to $\langle b_{i + 1}, f' \rangle$, for some facts $f, f'
  \in I \cup \{ \circ \}$. It follows from the construction of $D$ that there is an atom
  $R(\langle b_i, f\rangle, \langle b_{i + 1}, f'\rangle) \in D$ for all $f, f'$,
  except for $f = f' = R(a_i, a_{i + 1}) \in I$ (and $i + 1 \leq n$). 
  Assume for contradiction that there is an atom $R(a_i, a_{i + 1}) \in I$ and
  $f = f' = R(a_i, a_{i + 1})$. By definition of $g$ and $f = R(a_i, a_{i + 1})$ it follows that
  there is no fact $R(a, a'') \in J'$ with $(I_{i + 1}, a_{i + 1}) \to (J', a'')$, and it follows from
  definition of $g$ and $f' = R(a_i, a_{i + 1})$ that $(I_{i + 1}, a_{i + 1}) \to (J', a')$. A
  contradiction, since $R(a, a') \in J'$.
\end{proof}

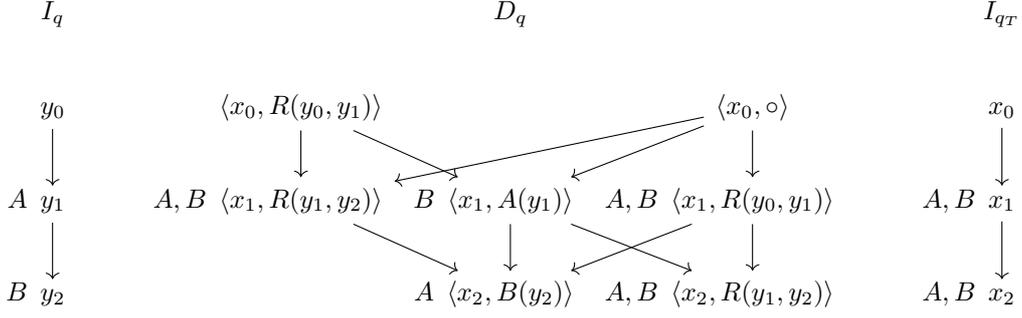
\begin{figure}
  \small
\[\begin{tikzcd}[label distance=-1.5mm, column sep = 1mm]
  I_q & & & D_q & & & I_{q_T} \\
	{y_0} & \quad & {\langle x_0, R(y_0, y_1) \rangle} && {\langle x_0, \circ \rangle} & \quad & {x_0} \\
	|[label=left:A]|{y_1} && |[label=left:{A,B}]|{\langle x_1, R(y_1, y_2) \rangle} & |[label=left:B]|{\langle x_1, A(y_1) \rangle} & |[label=left:{A, B}]|{\langle x_1, R(y_0, y_1) \rangle} && |[label=left:{A,B}]|{x_1} \\
	|[label=left:B]|{y_2} &&& |[label=left:A]|{\langle x_2, B(y_2) \rangle} & |[label=left:{A, B}]|{\langle x_2, R(y_1, y_2) \rangle} && |[label=left:{A,B}]|{x_2}
	\arrow[from=2-3, to=3-3]
	\arrow[from=2-3, to=3-4]
	\arrow[from=2-5, to=3-3]
	\arrow[from=2-5, to=3-4]
	\arrow[from=2-5, to=3-5]
	\arrow[from=3-3, to=4-4]
	\arrow[from=3-4, to=4-4]
	\arrow[from=3-5, to=4-4]
	\arrow[from=3-4, to=4-5]
	\arrow[from=3-5, to=4-5]
	\arrow[from=2-7, to=3-7]
	\arrow[from=3-7, to=4-7]
	\arrow[from=2-1, to=3-1]
	\arrow[from=3-1, to=4-1]
\end{tikzcd}\]
  \caption{The examples $(I_q, y_0)$ and $(I_{q_T}, x_0)$ are path examples.
  The instance $D_q$ is the result of the construction in the proof of Lemma~\ref{lem:path-dual}.
  The pair $(\{ (I_q, y_0)\},\{(D_q, \langle x_0, R(y_0, y_1) \rangle)\} )$ is a homomorphism duality relative to the example $(I_{q_T}, x_0)$.}
  \label{fig:example}
\end{figure}

Figure~\ref{fig:example} shows an example of the construction employed in the
proof of Lemma~\ref{lem:path-dual}. Note that the result of this construction
$D_q$ is not itself a path instance, and not even a tree.
In what follows, we  use $p$ to refer to a polynomial function $p \colon
\mathbb{N}^2 \to  \mathbb{N}$ that bounds the size of the data example $(D, d)$,
depending on the size of the examples $(I, a_0)$ and $(J, b_0)$, constructed according to
Lemma~\ref{lem:path-dual}.
We now build on Lemma~\ref{lem:path-dual} to  show that algorithms that produce strongly most-general fittings are not
sample-efficient PAC learning algorithms.

\begin{theorem}\label{thm:most-general-not-sample-efficient}
  Let $\Amf$ be a fitting algorithm for CQs that always produces a strongly most-general 
  fitting, if it exists. Then $\Amf$ is not a sample-efficient PAC learning algorithm.
\end{theorem}

\begin{proof}
  Assume to the contrary of what we aim to show that there is a sample-efficient
  PAC learning algorithm that produces a strongly most-general fitting CQ, if it exists,
  with associated polynomial function $m \colon \mathbb{R}^2\times \mathbb{N}^3 \rightarrow \mathbb{N}$
  as in Definition~\ref{def:efflearn}. We assume that the value of $m$ is at least 2, for any input. 

  Choose a signature $\mathcal{S}$ that contains the unary relation symbols $A$, $B$
  and a binary relation symbol $R$, $\delta = 0.5$, $\epsilon = 0.25$,
  and $n \in \mathbb{N}$ large enough so that
  $2^{n - 1} > m( \frac{1}{\delta}, \frac{1}{\epsilon}, |\mathcal{S}|, 3n, p(3n, 2n) )$.
  Without loss of generality we assume that $p(3n, 2n) \geq 3n$.

  As target CQ, we use
  \[
    q_T(x_0) \text{ :- } R(x_0, x_1), A(x_1), B(x_1), \dots, R(x_{n - 1}, x_n), A(x_n), B(x_n).
  \]
  Thus, $q_T$ is an $r$-path of length~$n$ in which every non-root node is labeled with $A$ and~$B$.
  We will use $(I_{q_T}, x_0)$ as a positive example for $q_T$. Note that
  $I_{q_T}$ is a path instance with $|I_{q_T}| = 3n \leq p(3n, 2n)$.

  Next, we construct instances that we use as negative examples.
  Define a set of CQs
  \[
    S = \{ q(y_0) \text{ :- } R(y_0, y_1), \alpha_1(y_1), \dots, R(y_{n - 1}, y_n), \alpha_n(y_n) \mid \alpha_i \in \{A, B\}\}.
  \]
  The CQs in $S$ resemble $q_T$, except that every node is labeled with only one
  of the unary relation symbols $A$ and $B$.
  For all elements $q$ of $S$ it holds that $q_T \subseteq q$ and $I_q$ is a path instance with $|I_q| = 2n$.
  In order to obtain the negative examples, we construct for each $q \in S$, the
  example $(D_q, a_q)$ such that $(\{(I_q, y_0)\}, \{(D_q, a_q)\})$ is a
  homomorphism duality relative to $(I_{q_T}, x_0)$.
  By definition of $p$, $|D_q| \leq p(3n, 2n)$.
  Since $(I_q, y_0) \to (I_{q_T}, x_0)$ for all $q \in S$, it follows from the definition of
  relativized homomorphism duality that $(I_{q_T}, x_0) \not \to (D_q, a_q)$.
  Hence all $(D_q, a_q)$ are negative examples for $q_T$.
  Figure~\ref{fig:example} shows examples of these instances for $n = 2$.

  The central properties of the chosen data examples are the following.
  For $q_1, q_2 \in S$, let $q_1 \land q_2$ denote the unary CQ that joins
  $q_1$ and $q_2$ at the answer variable $y_0$, but keeps all existential variables distinct.
  For $S' \subseteq S$, let $q_{S'} = \bigwedge_{q \in S} q$.

  \medskip\noindent\textit{Claim.}
  For every $S' \subseteq S$, $(\{(I_{q_{S'}}, y_0)\}, \{ (D_q, a_q) \mid q \in
  S'\})$ is a homomorphism duality relative to $(I_{q_T}, x_0)$.

  \smallskip\noindent\textit{Proof of claim.}
  Let $S'$ be a subset of $S$.
  Using the definition of homomorphism duality, we have to show that for all data 
  examples $(I, a)$ with $(I, a) \to (I_{q_T}, x_0)$, it holds that $(I_{q_{S'}}, y_0) \to (I,
  a)$ if and only if $(I, a) \not \to (D_q, a_q)$ for all $q \in S'$.

  Let $(I, a)$ be an example such that $(I, a) \to (I_{q_T}, x_0)$.
  First, we show that $(I_{q_S'}, y_0) \to (I, a)$ implies that $(I, a) \not\to
  (D_q, a_q)$ for all $q \in S'$. So assume that  $(I_{q_S'}, y_0) \to (I, a)$. 
 It suffices to show that $(I_{q_{S'}}, y_0) \not\to
  (D_q, a_q)$ for all $q \in S'$.
  Take any $q \in S'$. Since $(\{(I_q, y_0) \}, \{(D_q, a_q)\})$ is a homomorphism duality relative
  to $(I_{q_T}, x_0)$ and $(I_q, y_0)  \to (I_{q_T}, x_0)$, we have
  $(I_q, y_0)  \not\to (D_q, a_q)$. By construction of $q_{S'}$, this
  implies
   $(I_{q_{S'}}, y_0) \not\to (D_q, a_q)$.

  Next, we show that $(I_{q_S'}, y_0) \not\to (I, a)$ implies that 
  there is a $q \in S'$ such that $(I, a)  \to (D_q, a_1)$.
  Since $I_{q_{S'}}$ is the join of all $I_q$ with $q \in S'$ as the answer variable, 
  there must be a $q \in S'$ such that $(I_q, y_0) \not\to (I, a)$. 
  By definition of homomorphism dualities therefore
  $(I, a) \to (D_q, a_q)$, as required.
  This completes the proof of the claim.

  \medskip

  Let the probability distribution $P$ assign probability
  $\frac{1}{2}$ to the positive example $(I_{q_T}, x_0)$, probability
  $\frac{1}{2^{n + 1}}$ to all negative examples $(D_q, a_q)$, $q \in S$, and
  probability $0$ to all other data examples.
  Now assume that the algorithm is started on a collection of
  $m(\frac{1}{\delta},\frac{1}{\epsilon},|\Smc|,3n, p(3n, 2n))$
  data examples $E = (E^+, E^-)$ drawn according to~$P$ and labeled according to $q_T$. 
  Let $S' = \{ q \in S \mid (D_q, a_q)  \in E^-\} \subseteq S$.
  We argue that $q_{S'}$ fits $E$. Since $(I_{q}, y_0) \to (I_{q_T}, x_0)$ for all $q \in S$,
  it follows that $q_{S'}$ fits the positive examples in $E$ which are all of the form $(I_{q_T}, x_0)$.
  Now let $(D_q,a_q)$ be a negative example in $E$.
  We have $(I_{q}, y_0) \not\to (D_q, a_q)$ by the properties of
  homomorphism dualities and by
  construction of $q_{S'}$  
  it  follows that $q_{S'} \not\to (D_q, a_q)$, as required. 

  We next observe that if $(I_{q_T}, x_0) \in E^+$, then it follows from
  Theorem~\ref{thm:most-general-fitting-chara} and the claim that $q_{S'}$ is a
  strongly most-general fitting of $E$. 
  Hence, with probability $1 - \frac{1}{2^{|E|}} > 1 - \delta$ (since the value of
  $m$ and thus $|E|$ is at least~2), the algorithm
  $\mathbf{A}$ returns a CQ equivalent to $q_{S'}$.

  We argue that $q_{S'}$ classifies all negative examples $(D_q, a_q)$ with $q \in
  S \setminus S'$ incorrectly. To see this, let $q$ be a CQ from $S \setminus S'$.
  Note that $(I_q, y_0) \not \to (I_{q_{S'}}, y_0)$ by construction of~$q_{S'}$.
  Then, it follows from $(\{(I_q, y_0)\}, \{D_q, a_q\})$ being a homomorphism
  duality relative to $(I_{q_T}, x_0)$ and $(I_{q_{S'}}, y_0) \to (I_{q_T}, a_{q_T})$
  that $(I_{q_{S'}}, y_0) \to (D_q, a_q)$. Therefore all $(D_q, a_q)$ with $q \in S \setminus S'$
  are positive examples for $q_{S'}$.

  Since $|S \setminus S'| = 2^n - |S'| > 2^n - 2^{n - 1}$ by choice of $n$ and
  each of the  negative examples $(D_q, a_q)$ with $q \in
  S \setminus S'$ has probability $\frac{1}{2^{n + 1}}$ to be drawn from $P$, 
  we have $\mn{error}_{P, q_T}(q_{S'}) > \frac{2^n - 2^{n - 1}}{2^{n + 1}} = 1/4 = \epsilon$, a contradiction.
\end{proof}

\section{Most-Specific Fittings Preclude Sample-Efficient PAC Learning}

We show that fitting algorithms that always produce a most-specific fitting CQ
cannot be sample-efficient PAC learning algorithms. Note that a most-specific
fitting CQ is always guaranteed to exist provided that a fitting CQ exists 
at all~\cite{pods2023:extremal}. The proof of the following result is a 
refinement of a similar proof given in \cite{IJCAI23:SAT} for tree-shaped CQs.

\begin{theorem}\label{thm:most-specific-is-not-sample-efficient}
Let $\Amf$ be a fitting algorithm for CQs that always produces a most-specific
fitting CQ. 
Then $\Amf$ is not a sample-efficient PAC learning
algorithm.
\end{theorem}

\begin{proof}
  Assume to the contrary of what
  we aim to show that $\Amf$ is a sample-efficient PAC learning algorithm with
  associated polynomial function
  $m \colon \mathbb{R}^2 \times \mathbb{N}^3 \rightarrow \mathbb{N}$ as in
  Definition~\ref{def:efflearn}. Choose a schema $\Smc$ that contains a unary relation symbol $A$ and a binary relation symbol $R$,  
  set $q_T(x_0) \text{ :- } A(x_0)$,
  $\delta=\epsilon = 0.5$, and $n$ even and large enough such that
  \[ 
    \frac{1}{2}\binom{n}{n/2} > m\left(\frac{1}{\epsilon},\frac{1}{\delta},|\Smc|,1, p(2n, 2n) + 1\right).
  \]
  We next construct positive examples for $q_T$; negative examples are not used.
  Let $N$ denote the set of subsets of $\{1,\dots,n\}$ and let
  $N^{\frac{1}{2}}$ be defined likewise, but include only sets of
  cardinality exactly $n/2$. With every $S \in N$, we associate
  the path instance
  \[
     I_S = \{ R(b_0,b_1),\dots,R(b_{n-1},b_n) \} \cup \{ A(b_i) \mid i \in S \}.
  \]
  as well as the example  $(I'_S,a_0)$ that can be obtained
  by constructing an example $(J_S, a_0)$ such that
  $(\{(I_S, b_0)\}, \{(J_S, a_0)\})$ is a homomorphism duality relative to
  $(I_{\{1, \dots, n\}}, b_0)$ via the construction in the proof of Lemma~\ref{lem:path-dual}
  and then adding the fact $A(a_0)$.
  Due to this additional fact, every $(I'_S, a_0)$ is a positive data example
  for $q_T$ with $|I'_S| \leq p(2n, 2n) + 1$.
  
  Let $P$ be the probability distribution that assigns
  probability~$1/|N^{\frac{1}{2}}|$ to every $(I'_S,a_0)$
  with $S \in N^{\frac{1}{2}}$, and probability~$0$ to all other
  examples.
  Now assume that $\Amf$ is started on a collection of
  $m(\frac{1}{\epsilon},\frac{1}{\delta},|\Smc|, 1, p(2n, 2n) + 1)$ 
  data examples~$E$ drawn from $P$ and labeled according to $q_T$, and let $q_H$
  be the CQ that is output by $\Amf$. Since $E$ contains no negative examples, 
  Theorem~\ref{thm:most-specific-fitting-chara} implies that a most-specific
  fitting CQ exists. By the properties of~$\Amf$, $q_H$ must therefore be be
  a most-specific CQ that fits $E$.

  We argue that $q_H$ labels all data examples $(I'_S,a_0)$ with $S \in
  N^{\frac{1}{2}}$ that are not in the sample $E$ incorrectly. 
  Let $S$ be any element of $N^{\frac{1}{2}}$ with $(I'_S, a_0) \notin E$.
  We aim to show that $(I_{q_H}, a_{q_H}) \not\to (I'_S, a_0)$.
  Let $q_S$ be the canonical CQ of $(I_S, b_0)$.
  Since $(\{(I_S, b_0) \}, \{ (J, a_0) \})$ is a homomorphism duality relative
  to  $(I_{1, \dots, n}, b_0)$, and $I'_S = J \cup \{A(a_0) \}$, $(I_S, b_0)
  \not \to (I'_S, a_0)$. Hence, it suffices to show that $(I_S, b_0) \to
  (I_{q_H}, a_{q_H})$ or equivalently that $q_H \subseteq q_S$.

  Recall that $q_H$ is a most-specific fitting for $E$. Hence, it 
  suffices to show that $q_S$ is a fitting for $E$.
  Let $S'$ be any element of $N^{\frac{1}{2}}$ with $(I'_{S'}, a_0) \in E$. 
  Then $S \neq S'$ and thus $(I_{S'}, b_0) \not \to (I_{S}, b_0)$ by construction.
  Since $(I_{S}, b_0)  \to (I_{\{1, \dots, n\}}, b_0)$
  and $(\{(I_{S'}, b_0)\}, \{(J_{S'}, a_0)\})$ is a homomorphism duality
  relative to $(I_{\{1, \dots n\}}, b_0)$, 
  it follows that
  $(I_{S}, b_0) \to (J_{S'}, a_0)$ and hence
  $(I_{S}, b_0) \to (I'_{S'}, a_0)$. Therefore $(I'_{S'}, a_0)$ is a positive
  example for $q_S$ and we are done.
  
  \smallskip
  The definition of $P$ and the choice of $n$ now yield that with probability $1 > 1 - \delta$,
  \[
    \mn{error}_{P,q_T}(q_H) = \frac{|N^{\frac{1}{2}}| - |E|}{|N^{\frac{1}{2}}|} > 0.5,
  \]
  a contradiction.
\end{proof}

\bibliographystyle{plain}
\bibliography{bib}

\end{document}